\documentclass[submission,copyright,creativecommons,sharealike,noncommercial]{eptcs}
\providecommand{\event}{QPL 2017}
\pdfoutput=1

\usepackage{mathtools} 
\usepackage{amssymb} 
\usepackage{amsthm} 
\usepackage{stmaryrd} 

\usepackage{relsize} 
\usepackage{microtype} 
\usepackage{csquotes} 

\usepackage{hyperref} 
\usepackage[nocompress]{cite} 

\usepackage{graphicx} 
\usepackage[usenames,dvipsnames]{xcolor} 
\usepackage{tikz} 
\usepackage{circuitikz} 
\usetikzlibrary{
	arrows,
	shapes,
	decorations,
	intersections,
	backgrounds,
	positioning,
	circuits.ee.IEC
	}

\newcommand{\inlineQuote}[1]{\textquotedblleft #1\textquotedblright} 

\newcommand{\naturals}{\mathbb{N}} 
\newcommand{\rationals}{\mathbb{Q}} 
\newcommand{\reals}{\mathbb{R}} 
\newcommand{\complexs}{\mathbb{C}} 
\newcommand{\fRelCategory}{\operatorname{fRel}} 
\newcommand{\RMatCategory}[1]{#1\operatorname{-Mat}} 

\newcommand{\ket}[1]{\vert #1 \rangle} 
\newcommand{\bra}[1]{\langle #1 \vert} 
\newcommand{\braket}[2]{\langle #1 \vert #2 \rangle} 
\newcommand{\id}[1]{id_{#1}} 

\newcounter{theorem_c} 
\numberwithin{theorem_c}{section} 
\numberwithin{equation}{section} 
\theoremstyle{plain} 
	\newtheorem{theorem}[theorem_c]{Theorem}
	
	\newtheorem{lemma}[theorem_c]{Lemma}
	
        \newtheorem{counterexample}[theorem_c]{Counterexample}
\newtheoremstyle{exampstyle}
  {2mm} 
  {2mm} 
  {\itshape} 
  {} 
  {\bfseries} 
  {.} 
  {.5em} 
  {} 
\theoremstyle{exampstyle}
	\newtheorem{definition}[theorem_c]{Definition}

\tikzstyle{every picture}=[baseline=-0.25em,scale=0.5]
\pgfdeclarelayer{edgelayer}
\pgfdeclarelayer{nodelayer}
\pgfsetlayers{edgelayer,nodelayer,main}
\tikzstyle{labelnode}=[fill=white]
\tikzstyle{none}=[inner sep=0pt]

\DeclareRobustCommand{\coprod}{\mathop{\text{\fakecoprod}}}
\newcommand{\fakecoprod}{%
  \sbox0{$\prod$}%
  \smash{\raisebox{\dimexpr.9625\depth-\dp0}{\scalebox{1}[-1]{$\prod$}}}%
  \vphantom{$\prod$}%
}
\newcommand{\fatcat}[2]{\coprod_{#1}{#2}}

 \def\copyrightholders{Coecke, Genovese, Gogioso, Marsen, \& Piedeleu} 

\title{Uniqueness of Composition\\in Quantum Theory and Linguistics}
\author{
	Bob Coecke\\
	University of Oxford \\
	\texttt{bob.coecke@cs.ox.ac.uk}
	\and
	Fabrizio Genovese \\
	University of Oxford \\
	\texttt{fabrizio.genovese@cs.ox.ac.uk}
	\and
	Stefano Gogioso\\
	University of Oxford \\
	\texttt{stefano.gogioso@cs.ox.ac.uk}
	\and
	Dan Marsden\\
	University of Oxford \\
	\texttt{dan.marsden@cs.ox.ac.uk}
	\and
	Robin Piedeleu\\
	University of Oxford \\
	\texttt{robin.piedeleu@cs.ox.ac.uk}
}

\def\titlerunning{Uniqueness of Composition in Quantum Theory and Linguistics}
\def\authorrunning{B. Coecke, F. Genovese, S. Gogioso, D. Marsden \& R. Piedeleu}

\begin{document}

\maketitle

\begin{center}\textit{Authors of this work are listed in alphabetical order.}\end{center}

\begin{abstract}
	We derive a uniqueness result for non-Cartesian composition of systems in a large class of process theories, with important implications for quantum theory and linguistics.
	Specifically, we consider theories of wavefunctions valued in commutative involutive semirings---as modelled by categories of free finite-dimensional modules---and we prove that the only bilinear compact-closed symmetric monoidal structure is the canonical one (up to linear monoidal equivalence). Our results apply to conventional quantum theory and other toy theories of interest in the literature, such as real quantum theory, relational quantum theory, hyperbolic quantum theory and modal quantum theory. In computational linguistics they imply that linear models for categorical compositional distributional semantics (DisCoCat)---such as vector spaces, sets and relations, and sets and histograms---admit an (essentially) unique compatible pregroup grammar. 
\end{abstract}

\noindent Cartesian products have played a very important role in category theory since its very early days. In the past few years, however, a number of theories have been developed which have manifestly non-Cartesian monoidal tensor products, such as Kelly's compact closed tensor symmetric monoidal structure \cite{Kelly, KellyLaplaza}. Prominent examples of such theories include (but are by no means limited to):
\begin{itemize}
  \item categorical quantum mechanics (CQM) \cite{AC1, CKBook}, which re-formulates the structural backbone of quantum theory in purely compositional terms; 
  \item categorical compositional distributional models of meaning (DisCoCat) \cite{Coecke2010}, which provide an algorithm to compute meaning of phrases and sentences in natural language processing, given the meaning of individual words and the underlying grammatical structure. 
\end{itemize}
In the case of CQM, the symmetric monoidal structure models the notion of composite quantum systems, while the compact closed structure models maximal entanglement and state-operator duality. In the case of DisCoCat, the compact closed symmetric monoidal structure model the effects of grammar, e.g. by mediating how a transitive verb \inlineQuote{consumes} its object and a subject in order to produce a sentence.  

When a category is interpreted as a theory of processes (i.e. is seen as a \textit{process theory}), having a Cartesian tensor product amounts to a very restrictive \textit{property}: it means that the state of a composite system is specified entirely by states of its component parts. In the diagrammatic language for symmetric monoidal categories (see e.g.~\cite{CKBook}), for example, the states of a bipartite system in a Cartesian theory always take the following, separable form: 
\begin{equation}
\begin{tikzpicture}[scale=0.6]
	\begin{pgfonlayer}{nodelayer}
		\node [style=none] (0) at (-7.25, -0) {};
		\node [style=none] (1) at (-2.25, -0) {};
		\node [style=none] (2) at (-4.75, -1.5) {};
		\node [style=none] (3) at (-3.75, -0) {};
		\node [style=none] (4) at (-3.75, 3) {};
		\node [style=none] (5) at (1.75, -1.5) {};
		\node [style=none] (6) at (0.25, -0) {};
		\node [style=none] (7) at (3.25, -0) {};
		\node [style=none] (8) at (1.75, -0) {};
		\node [style=none] (9) at (1.75, 3) {};
		\node [style=none] (10) at (5.5, 3) {};
		\node [style=none] (11) at (5.5, -0) {};
		\node [style=none] (12) at (5.5, -1.5) {};
		\node [style=none] (13) at (7, -0) {};
		\node [style=none] (14) at (4, -0) {};
		\node [style=none] (15) at (-1, -0) {=};
		\node [style=none] (16) at (-5.75, -0) {};
		\node [style=none] (17) at (-5.75, 3) {};
	\end{pgfonlayer}
	\begin{pgfonlayer}{edgelayer}
		\draw (0.center) to (1.center);
		\draw (1.center) to (2.center);
		\draw (2.center) to (0.center);
		\draw (3.center) to (4.center);
		\draw (6.center) to (7.center);
		\draw (7.center) to (5.center);
		\draw (5.center) to (6.center);
		\draw (8.center) to (9.center);
		\draw (14.center) to (13.center);
		\draw (13.center) to (12.center);
		\draw (12.center) to (14.center);
		\draw (11.center) to (10.center);
		\draw (16.center) to (17.center);
	\end{pgfonlayer}
\end{tikzpicture}
\end{equation}
The use of the word \textit{property} takes here a rigorous, formal meaning: the Cartesian tensor product for a category is defined by certain limits, and hence when it exists it is essentially unique. Conversely, symmetric monoidal structure is in general not a property in this sense: a generic category might admit essentially inequivalent tensor products, even when the additional requirements of symmetric and compact closure are enforced. In this sense, compact closed symmetric monoidal structure is in general a \textit{structure} which one imposes on a category, rather than a property that the category possesses.

Thus said, a conceptual analysis of process theories admitting compact closed symmetric monoidal structure (which we will refer to as \textit{compact closed} process theories) seems to indicate that the latter \textit{should} behave as a property, rather than just being a structure: intuitively, a Cartesian theory is one in which composite systems are non-interacting (or \inlineQuote{minimally} interacting), while a compact closed process theory is one in which composite systems are \inlineQuote{maximally} interacting. To see the latter point, consider the yanking equations that define compact closure in the diagrammatic formalism, in terms of \textit{cups} and \textit{caps}:
\begin{equation}
\begin{tikzpicture}[scale=0.7]
	\begin{pgfonlayer}{nodelayer}
		\node [style=none] (0) at (-0.25, -1) {};
		\node [style=none] (1) at (-0.25, 2.5) {};
		\node [style=none] (2) at (-4.25, -2.5) {};
		\node [style=none] (3) at (-4.25, 0.75) {};
		\node [style=none] (4) at (-2.25, 0.75) {};
		\node [style=none] (5) at (-2.25, -1) {};
		\node [style=none] (6) at (-1.75, 2) {};
		\node [style=none] (7) at (-1.75, 0.25) {};
		\node [style=none] (8) at (-4.75, 0.25) {};
		\node [style=none] (9) at (-4.75, 2) {};
		\node [style=none] (10) at (0.25, -0.5) {};
		\node [style=none] (11) at (-2.75, -0.5) {};
		\node [style=none] (12) at (-2.75, -2.25) {};
		\node [style=none] (13) at (0.25, -2.25) {};
		\node [style=none] (14) at (1.5, -0) {$=$};
		\node [style=none] (15) at (2.75, 2.5) {};
		\node [style=none] (16) at (2.75, -2.5) {};
	\end{pgfonlayer}
	\begin{pgfonlayer}{edgelayer}
		\draw (0.center) to (1.center);
		\draw (2.center) to (3.center);
		\draw (5.center) to (4.center);
		\draw [bend right=90, looseness=1.50] (4.center) to (3.center);
		\draw [bend right=90, looseness=1.50] (5.center) to (0.center);
		\draw [dotted] (7.center) to (8.center);
		\draw [dotted] (8.center) to (9.center);
		\draw [dotted] (9.center) to (6.center);
		\draw [dotted] (6.center) to (7.center);
		\draw [dotted] (13.center) to (12.center);
		\draw [dotted] (12.center) to (11.center);
		\draw [dotted] (11.center) to (10.center);
		\draw [dotted] (10.center) to (13.center);
		\draw (16.center) to (15.center);
	\end{pgfonlayer}
\end{tikzpicture}
\end{equation}
Because an appropriate composition of cup and cap results in the identity---which is the process of \inlineQuote{preserving everything that can be said about a system}---they can be seen to realise a maximal flow of information between two systems. In this sense of \inlineQuote{mediating maximal interaction}, one would expect compact closed symmetric monoidal structure to behave as a property. Hence the question that this paper sets out to answer: is compact closed symmetric monoidal structure essentially unique in those process theories of interest in CQM and DisCoCat? In physical terms, is there a unique notion of maximal entanglement in a theory? In linguistic terms, is the manner in which meaning composes through the grammar unique in a given semantic model? To put it another way: are the philosophical, intuitive arguments drafted above backed up by suitably rigorous mathematical results?

In Section~\ref{section:bobwrong}, we show that our question has non-trivial content, by providing examples of categories with inequivalent compact closed symmetric monoidal structure. In Section~\ref{section:bobright}, we show that compact closed symmetric monoidal structure is indeed a property for a large family of process theories of interest quantum theory and linguistics. In Section~\ref{section:NLP}, finally, we comment on the impact of our observation for compositional models of natural language.

\section{Compactness is not a property...}
\label{section:bobwrong}
In this Section, we see that compact closed symmetric monoidal structure on a category is, in general, very far from being essentially unique. Our first counterexample to essential uniqueness comes from certain (degenerate) Lambek pregroups~\cite{Lambek1997}, of interest in natural language processing.
\begin{lemma}
  \label{lem:discreteabeliangroups}
  The discrete category on the underlying set $X$ of an abelian group~$G:=(X,\times,e)$ carries the structure of a compact closed symmetric monoidal category, with the monoidal tensor product given by the group multiplication, tensor unit the group identity~$e$, and dual objects given by the group inverse.
\end{lemma}
\begin{proof}
  It is trivial to observe that the group multiplication extends to a bi-functor, as all the morphisms in the category are identities. The tensor clearly gives a strict symmetric monoidal structure, and compact closed structure amounts to the observation that~$g \times g^{-1} = e$ (and the snake equations follow automatically).
\end{proof}
\begin{lemma}
  \label{lem:basiciso}
  For abelian groups~$G_1$ and~$G_2$, there is a monoidal equivalence between the corresponding discrete symmetric monoidal categories from Lemma~\ref{lem:discreteabeliangroups} if and only if~$G_1$ and~$G_2$ are isomorphic groups.
\end{lemma}
\begin{proof}
  Any equivalence of discrete categories must be an isomorphism, and any strong monoidal functor between discrete monoidal categories must be a strict monoidal functor. Therefore a strong monoidal equivalence between these categories must amount exactly to a group isomorphism.
\end{proof}
\begin{counterexample}
  \label{counter:basic}
  There are two non-isomorphic abelian group structures on a four element set, namely the cyclic group $\mathbb{Z}_4$ and the Klein four-group $\mathbb{Z}_2 \times \mathbb{Z}_2$. Therefore the discrete category with four objects carries two monoidally inequivalent compact closed symmetric monoidal structures.
\end{counterexample}
Counterexample~\ref{counter:basic} is somewhat unsatisfactory, as the categories in question are discrete. We can address this using the following construction, which \inlineQuote{fattens up} the Homsets. Let be $X$ be a set and~$M$ a commutative monoid. Consider the~$X$-fold coproduct of~$M$, i.e. the category~$\fatcat{X}{M}$ having the elements of~$X$ as objects and Homsets specified as follows:
\begin{equation*}
  \fatcat{X}{M}\big(g,g'\big) =
  \begin{cases}
    M &\text{ if } g = g'\\
    \emptyset &\text{ otherwise}
  \end{cases}
\end{equation*}
Fixing an abelian group structure~$G:=(X,\times,e)$ on the set $X$ endows the category $\fatcat{X}{M}$ with the structure of a compact closed symmetric monoidal category, which we denote by $\fatcat{G}{M}$: the monoidal unit is~$e$, the tensor product is given by the group multiplication, and the cups and caps are given by the identity element in~$M$.
\begin{lemma}
  \label{lem:fatdisconnectediso}
  For two abelian groups~$G_1$ and~$G_2$, and a commutative monoid~$M$ with no non-trivial inverses (e.g. a free monoid), there is a strong monoidal equivalence between the corresponding monoidal categories~$\fatcat{G_1}{M}$ and~$\fatcat{G_2}{M}$ if and only if~$G_1$ and~$G_2$ are isomorphic groups.
\end{lemma}
\begin{proof}
  Any equivalence of categories must be an isomorphism, as the categories are disconnected. Any strong monoidal functor between such categories must be strict monoidal functor (there are no non-trivial isomorphisms), and in particular its action on the objects must induce an isomorphism of groups.
\end{proof}
\begin{counterexample}
  \label{counter:fatdisconnected}
  The free monoid $M:=\{a,b\}^*$ on the two element set~$\{a,b\}$ has no non-trivial inverses. By Lemma~\ref{lem:fatdisconnectediso} the 4-fold coproduct of~$\{a,b\}^*$ carries two monoidally inequivalent compact closed symmetric monoidal structures, induced by the cyclic group $\mathbb{Z}_4$ and the Klein four-group $\mathbb{Z}_2 \times \mathbb{Z}_2$.
\end{counterexample}
A similar construction can be used to produce a connected family of counterexamples, for all sets $X$ and commutative monoids $M$ with no non-trivial inverses: objects are the elements of $X$ as before, but now all Homsets are taken to be the same commutative monoid $M$, with composition of morphisms defined to be product in $M$. Any two non-isomorphic abelian groups structures $G_1 := (X,\times_1,e_1)$ and $G_2 := (X,\times_2,e_2)$ on the same underlying set $X$ will endow the category with two compact closed symmetric monoidal structures---where tensor product of morphisms is defined to be product in $M$---which are not monoidally isomorphic. Unfortunately, this counterexample cannot be strengthened to monoidal inequivalence, but a thorough investigation of this issue is left to future work.

\section{...except in those cases where it is a property.}
\label{section:bobright}
Consider a commutative semiring $S$ with involution (used for conjugation, duals and dagger), and define the dagger compact category $\RMatCategory{S}$ of free finite-dimensional $S$-modules as follows.
\begin{enumerate}
 	\item[(i)] The objects of $\RMatCategory{S}$ are the free finite-dimensional $S$-modules, in the form $S^X$ for finite sets $X$ (we will denote by $\ket{x}_{x \in X}$ the standard orthonormal basis of $S^X$).
 	\item[(ii)] The morphisms of $\RMatCategory{S}$ are the $S$-linear maps, which we can think of as $S$-valued matrices.
 	\item[(iii)] The tensor product is the usual symmetric tensor product of $S$-modules, and its action on morphisms is given by the Kronecker product of matrices.
 	\item[(iv)] The dagger of a matrix is given by its transpose conjugate (with respect to the involution $^\ast$ of $S$), and in particular the (possibly degenerate) inner product of two vectors $\ket{\phi} := \sum_x \phi_x \ket{x}$ and $\ket{\psi} := \sum_x \psi_x \ket{x}$ in the same $S$-module $S^X$ is given by:
 	\begin{equation}
 		\braket{\phi}{\psi} := \sum_{x} \phi_x^\ast \psi_x
 	\end{equation}
 	\item[(v)] The dual object $(S^X)^\ast$ is the free finite-dimensional $S$-module of $S$-linear maps $S^X \rightarrow S$, and we denote by $\ket{x^\ast}_{x \in X}$ the orthonormal basis of $(S^X)^\ast$ given by the $S$-linear maps $\bra{x}: S^X \rightarrow S$ for all $x \in X$. More in general, we denote by $\ket{\psi^\ast}$ the state in $(S^X)^\ast$ corresponding to the $S$-linear map $\bra{\psi}: S^X \rightarrow X$ (note that when $\ket{\psi} = \sum_x \psi_x \ket{x}$ we necessarily have $\ket{\psi^\ast} = \sum_x \psi_x^\ast \ket{x^\ast}$).
 	\item[(vi)] The compact closed structure is given by considering the caps $\epsilon_{S^X} : S^X \otimes (S^X)^\ast \rightarrow S$ and cups $\eta_{S^X}: S \rightarrow (S^X)^\ast \otimes S^X$ defined as follows:
 	\begin{equation}
 		\epsilon_{S^X} := \sum_{x \in X} \bra{x} \otimes \bra{x^\ast} \hspace{3cm} \eta_{S^X} := \sum_{x \in X} \ket{x^\ast} \otimes \ket{x}
 	\end{equation}
 	It is easy to check, for example, that $\epsilon_{S^X} \circ (\ket{\phi^\ast} \otimes \ket{\psi})$ gives the inner product $\braket{\phi}{\psi}$:
 	\begin{equation}
 		\epsilon_{S^X} \circ (\ket{\phi^\ast} \otimes \ket{\psi}) = (\sum_{x \in X} \bra{x} \otimes \bra{x}) \circ (\sum_x \phi_x^\ast \ket{x} \otimes  \psi_x\ket{x}) = \sum_x \phi_x^\ast \psi_x
 	\end{equation}
  In particular, the category $\RMatCategory{S}$ is enriched in itself, with morphisms $S^X \rightarrow S^Y$ carrying the natural structure of the free finite-dimensional $S$-module $(S^X)^\ast \otimes S^Y$. Composition and tensor product are $S$-bilinear bi-functors\footnote{This includes the statement that associators and unitors for the tensor product satisfy appropriate $S$-linearity conditions.}, and the dagger is an $S$-linear endofunctor.
 	\item[(vii)] The category $\RMatCategory{S}$ has biproducts given by Cartesian products of objects, and $S$-linear structure given by matrix addition; both tensor and dagger respect the biproducts and the $S$-linear structure.
\end{enumerate}
Categories in the form $\RMatCategory{S}$ are used to model quantum theory (for $S:=\complexs$, with complex conjugation as involution) and many other quantum-like theories of interest in the literature. Examples include: real quantum theory (for $S:=\reals$, with the identity as involution); relational quantum theory (for $S$ the booleans, with the identity as involution); hyperbolic quantum theory (for $S$ the split-complex numbers, with split-complex conjugation as involution); $p$-adic quantum theory (for $S$ a quadratic extension of the $p$-adic complex numbers, with field-theoretic conjugation as involution); modal quantum theory (for $S$ a finite field). For a detailed description of such quantum-like theories, see \cite{Gogioso2017}. 

The examples above suggest that the process theories modelled by $\RMatCategory{S}$ categories have very precise physical content. As a consequence, one could intuitively expect that the notion of parallel composition of processes induced by the tensor product should arise in a natural, essentially unique way, i.e. that it should be a property of these categories. Luckily, we are able to show that this is indeed the case.

\begin{theorem}\label{thm_freeFdSemimodulesRlinearUniqueness}
Let $S$ be a commutative semiring with involution. There is a unique $S$-bilinear compact-closed symmetric monoidal structure on $\RMatCategory{S}$ (up to $S$-linear monoidal equivalence).
\end{theorem}
\begin{proof}
Without loss of generality, we work with the skeletal, self-dual version of $\RMatCategory{S}$ which has the natural numbers as objects, with tensor product $n \otimes m := n \times m$ and duals $n^\ast := n$. We will construct a monoidal isomorphism on the natural numbers version of $\RMatCategory{S}$, that will subsequently lift to a monoidal equivalence on the full version of $\RMatCategory{S}$. 

We denote by $(\RMatCategory{S},\otimes,S,^\ast)$ the compact closed symmetric monoidal structure defined above, with caps $\epsilon_{n} : n \otimes n \rightarrow 1$ and cups $\eta_{n}: 1 \rightarrow n \otimes n $ (because we have $n^\ast = n$). We denote by $(\RMatCategory{S},\odot,J,^\circ)$ some other symmetric monoidal closed structure on $\RMatCategory{S}$, where $\odot$ is $S$-bilinear, and we assume that $(\RMatCategory{S},\odot,J,^\circ)$ is compact closed with caps $\bar{\epsilon}_{n} : n \odot n^\circ \rightarrow 1$ and cups $\bar{\eta}_{n}: 1 \rightarrow n^\circ \odot n$.

If $S$ is a commutative semiring with involution, then compact closure together with $S$-bilinearity gives an $S$-linear isomorphism between $S^{n \times m}$ (the $S$-linear maps $n \rightarrow m$) and $S^{J \times (n^\circ \odot m)}$ (the $S$-linear maps $J \rightarrow (n^\circ \odot m)$). As we are working with free finite-dimensional $S$-modules, we obtain  the following equation, valid for all $n,m \in \naturals$:
\begin{equation}\label{eqn_tensorOnObjects}
	J \times (n^\circ \odot m) = n \times m
\end{equation}
The RHS has to be divisible by $J$ for all $n,m \in \naturals$, so we immediately get that $J=1$, and we are left with the equation $n^\circ \odot m = n \times m$. But by taking $m=J$ we also get $n^\circ = n^\circ \odot J = n \times J = n\times 1 = n$, and hence we are left with $n \odot m = n \times m$, proving that $(\otimes, 1)$ and $(\odot, J)$ coincide on objects. 

Now we show that we can construct a monoidal isomorphism $F: (\RMatCategory{S},\otimes,S,^\ast) \rightarrow (\RMatCategory{S},\odot,J,^\circ)$, which is the identity on objects. For each finite prime $p$, let $(\ket{j;p})_{j=1}^p$ be the standard orthonormal basis for the object $p$ (corresponding to the free finite-dimensional $S$-module $S^p$). For each natural number $n \geq 2$, decompose $n$ into prime factors $n = p_{1;n} \cdot ... \cdot p_{K_n;n}$, making the decomposition unique by picking primes in non-decreasing order. Define an orthonormal basis for the object $n = p_{1;n} \otimes ... \otimes p_{K_n;n}$ (corresponding to the free finite-dimensional $S$-module $S^n \cong S^{p_{1;n}} \otimes ... \otimes S^{p_{K_n;n}}$) as follows:
\begin{equation}
\Big( \, \ket{\underline{j};n} :=  \bigotimes\limits_{k=1}^{K_n} \ket{j_k;p_{k;n}}  \, \Big)_{ \underline{j} \,\in\, \prod_{k=1}^{K_n} \{1,...,p_{k;n}\}}
\end{equation}
We will first define the functor $F$ on states and effects, and then we will extend it to a full functor by imposing $S$-linearity. The definition on states is given by $F \Big( \, \ket{\underline{j};n} \Big) :=  \bigodot\limits_{k=1}^{K_n} \ket{j_k;p_{k;n}}
$,\vspace{-3mm} while its dual on effects is given by $F \Big( \, \bra{\underline{j};n} \Big) :=  \bigodot\limits_{k=1}^{K_n} \bra{j_k;p_{k;n}}$. Note that $\odot$ need not respect the dagger structure of $(\RMatCategory{S},\otimes,S,^\ast)$, in which case $F$ need not be a dagger functor. On scalars, we necessarily have $a \odot b = a \circ b = a \otimes b$, without needing to invoke any $S$-linearity requirement; hence we can define $F$ to be the identity on scalars. Finally, we extend $F$ to arbitrary morphisms by functoriality and $S$-linearity:
\begin{equation}\label{eqn_monoidalIsomorphism}
F \Big( \, \sum_{\underline{i}}\sum_{\underline{j}} f_{\underline{j}\underline{i}} \ket{\underline{j};m}  \bra{\underline{i};n} \, \Big) :=   \sum_{\underline{i}}\sum_{\underline{j}} f_{\underline{j}\underline{i}} \Big(\bigodot\limits_{k=1}^{K_m} \ket{j_k;p_{k;m}}\Big) \Big(\bigodot\limits_{k=1}^{K_n} \bra{i_k;p_{k;n}}\Big)
\end{equation}
Equation \ref{eqn_monoidalIsomorphism} defines an $S$-linear monoidal functor $F: (\RMatCategory{S},\otimes,S,^\ast) \rightarrow (\RMatCategory{S},\odot,J,^\circ)$ which is the identity on objects. Now observe that $\odot$ must respect identities, and hence that we get the following resolution of the identity for every object $n$: 
\begin{align}\label{eqn_resolutionIdentity}
\id{n} = \sum_{\underline{i}}\sum_{\underline{j}} \Big(\bigodot\limits_{k=1}^{K_n} \ket{j_k;p_{k;n}}\Big) \Big(\bigodot\limits_{k=1}^{K_n} \bra{i_k;p_{k;n}}\Big)
\end{align}
Because of Equation \ref{eqn_resolutionIdentity}, all morphisms $g: n \rightarrow m$ can be written in matrix form as follows:
\begin{align}
g = \id{m} \circ g \circ \id{n} = \sum_{\underline{i}}\sum_{\underline{j}} g_{\underline{j}\underline{i}} \Big(\bigodot\limits_{k=1}^{K_m} \ket{j_k;p_{k;)}}\Big) \Big(\bigodot\limits_{k=1}^{K_n} \bra{i_k;p_{k;n}}\Big) 
\end{align}
Given the matrix form above, a monoidal inverse $F^{-1}$ for the monoidal functor $F$ is defined by:
\begin{equation}
F^{-1} \Big( \, \sum_{\underline{i}}\sum_{\underline{j}} g_{\underline{j}\underline{i}} \Big(\bigodot\limits_{k=1}^{K_m} \ket{j_k;p_{k;m}}\Big) \Big(\bigodot\limits_{k=1}^{K_n} \bra{i_k;p_{k;n}}\Big) \, \Big) := \sum_{\underline{i}}\sum_{\underline{j}} g_{\underline{j}\underline{i}} \ket{\underline{j};m}  \bra{\underline{i};n}
\end{equation}
Hence $F$ as defined by Equation \ref{eqn_monoidalIsomorphism} is a monoidal isomorphism, as required. In the full version of $\RMatCategory{S}$, this means that the tensor product part of the compact closed symmetric monoidal structure is essentially unique; however, it does not say anything about the cups and caps providing the compactness. 

It is well-known that cups and caps for a given symmetric monoidal structure are unique up to natural isomorphism \cite{Selinger2009}. However, the traditional result does not guarantee the existence of an $S$-linear isomorphism, i.e. one compatible with the linear structure of the theory which underpinned the uniqueness result for the symmetric monoidal structure in the first place. To show that cups and caps are essentially unique up to one such $S$-linear natural isomorphism, we consider an equivalent, extended version of $\RMatCategory{S}$, in which objects are no longer self-dual. Objects are given by pairs $n:=(n^+,n^-)$ of natural numbers (where $n^+$ and $n^-$ are either both zero or both non-zero), tensor product is given by $n \otimes m := (n^+\times m^+,n^- \times m^-)$, tensor unit is $(1,1)$, and duals are given by $n^\ast:= (n^-,n^+)$. The object $n$ labels the $S$-module $S^{n^+} \otimes (S^{n^-})^\ast$. In the light of our previous result on the skeletal, self-dual version,  we are also free to set $n \odot m := n \otimes m$ and $n^\circ := n^\ast$ in the non self-dual version, since we are working up to monoidal equivalence with respect to the full version. 

Consider the same orthonormal bases $(\ket{\underline{j};n})_{\underline{j}}$ and co-bases $(\bra{\underline{j};n})_{\underline{j}}$ the we used in the self-dual case. Let $(\ket{\underline{j};n;\ast})_{\underline{j}}$ and $(\ket{\underline{j};n;\circ})_{\underline{j}}$ be the corresponding dual bases, given by $(\otimes,^\ast)$ and $(\odot,^\circ)$ respectively. Similarly, let $(\bra{\underline{j};n;\ast})_{\underline{j}}$ and $(\bra{\underline{j};n;\circ})_{\underline{j}}$ be the corresponding dual co-bases, again given by $(\otimes,^\ast)$ and $(\odot,^\circ)$ respectively. To conclude our proof, we lift the monoidal isomorphism $F$ defined above on the skeletal self-dual version of $\RMatCategory{S}$ to a monoidal isomorphism $\bar{F}: (\RMatCategory{S},\otimes,S,^\ast) \rightarrow (\RMatCategory{S},\odot,J,^\circ)$ in the extended, non self-dual version of $\RMatCategory{S}$. On objects, $\bar{F}$ is defined to be the identity, just as in the self-dual case. By using the dual bases and co-bases, morphisms $f:(n^+,n^-) \rightarrow (m^+,m^-)$, seen as $S$-linear maps $f: S^{n^+} \otimes (S^{n^-})^\ast \rightarrow S^{m^+} \otimes (S^{m^-})^\ast$, can be expressed in matrix form as follows:
\begin{equation}
\sum_{\underline{a},\underline{b}}\sum_{\underline{c},\underline{d}} f_{(\underline{c},\underline{d})(\underline{a},\underline{b})} \Big(\ket{\underline{c};m^+} \otimes \ket{\underline{d};m^-;\ast}\Big) \Big( \bra{\underline{a};n^+} \otimes \bra{\underline{b};n^-;\ast} \Big)
\end{equation}
The functor $\bar{F}$ can then be defined on morphisms as follows:
\begin{align}
&\bar{F}\Big(\sum_{\underline{a},\underline{b}}\sum_{\underline{c},\underline{d}} f_{(\underline{c},\underline{d})(\underline{a},\underline{b})} \Big(\ket{\underline{c};m^+} \otimes \ket{\underline{d};m^-;\ast}\Big) \Big( \bra{\underline{a};n^+} \otimes \bra{\underline{b};n^-;\ast} \Big)\Big) \nonumber\\
&\hspace{2cm}:= \sum_{\underline{a},\underline{b}}\sum_{\underline{c},\underline{d}} f_{(\underline{c},\underline{d})(\underline{a},\underline{b})} \Big(\ket{\underline{c};m^+} \odot \ket{\underline{d};m^-;\circ}\Big) \Big( \bra{\underline{a};n^+} \odot \bra{\underline{b};n^-;\circ} \Big)
\end{align}
The functor $\bar{F}$ is evidently an $S$-linear monoidal isomorphism. Furthermore, it is immediate to check that $\bar{F}$ sends the cups and caps of $(\RMatCategory{S},\otimes,S,^\ast)$ to the cups and caps of $(\RMatCategory{S},\odot,J,^\circ)$. As a consequence, we conclude that the entire compact closed symmetric monoidal structure is essentially unique in the full version of $\RMatCategory{S}$, up to the desired $S$-linear monoidal equivalence.
\end{proof}

\section{The DisCoCat point of view}
\label{section:NLP}
Categorical compositional distributional semantics (DisCoCat) provides a framework to talk about the compositional structure of natural language \cite{Coecke2010}, and has been remarkably successful in various fields of Natural Language Processing, often outperforming competing approaches\cite{Grefenstette2011, Kartsaklis2014}. In this Section, we briefly summarise how categorical compositional distributional models of meaning work, and discuss the importance of our uniqueness result in their regard. 

In the DisCoCat formalism, the grammar of natural language is almost always mathematically formalised by using \textit{Lambek's pregroup grammars} \cite{Lambek1997} (although other categorical grammars~\cite{Coecke2013} have also proven to be effective). Pregroup grammars amount to the imposition of a pregroup structure on a set $P$ of grammatical types. 
\begin{definition}\label{pregroup}
A pregroup $(P,~\leq,~\cdot,~1,~(-)^l,~(-)^r)$ is a partially ordered monoid $(P, \leq, \cdot, 1)$, where each element $p\in P$ has a \inlineQuote{left adjoint} $p^l$ and a \inlineQuote{right adjoint} $p^r$ satisfying the following condition:
$p^l\cdot p \leq 1 \leq p\cdot p^l$ and $p\cdot p^r \leq 1 \leq p^r \cdot p$.
\end{definition}
\noindent A pregroup $P$ can be interpreted as monoidal poset category. The left and right adjoints then endow it with the structure of a thin, non-symmetric compact closed category (sometimes called autonomous, or rigid, in the literature), with cups and caps given by the inequalities in Definition \ref{pregroup} (see Figure~\ref{fig:directedcupcaps} below). 

\begin{figure}[h]
\centering
\begin{tikzpicture}[scale=2]
\begin{pgfonlayer}{nodelayer}
\node [style=none] (0) at (-3, 0.5) {$1 \leq p^r \cdot p$};
\node [style=none] (1) at (-1, 0.5) {};
\node [style=none] (2) at (3, 0.5) {$1 \leq p \cdot p^l$};
\node [style=none] (3) at (3, -0.3) {$p \cdot p^r \leq 1$};
\node [style=none] (4) at (-1, -0.5) {};
\node [style=none] (5) at (1, -0.5) {};
\node [style=none] (6) at (2, -0.5) {};
\node [style=none] (7) at (1, 0.5) {};
\node [style=none] (8) at (-2, 0.5) {};
\node [style=none] (9) at (-3, -0.3) {$p^l \cdot p \leq 1$};
\node [style=none] (10) at (2, 0.5) {};
\node [style=none] (11) at (-2, -0.5) {};
\end{pgfonlayer}
\begin{pgfonlayer}{edgelayer}
\draw [thick, >->, bend left=90, looseness=1.25] (1.center) to (8.center);
\draw [thick, >->, bend right=90, looseness=1.25] (4.center) to (11.center);
\draw [thick, >->, bend right=90, looseness=1.25] (7.center) to (10.center);
\draw [thick, >->, bend left=90, looseness=1.25] (5.center) to (6.center);
\end{pgfonlayer}
\end{tikzpicture}
\caption{The directed cups and caps in a pregroup grammar, given by the corresponding inequalities.}\label{fig:directedcupcaps}
\end{figure}
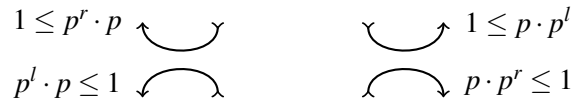

We can construct a grammar from a set of generators~$\{n, s\}$ for the noun and sentence types, with higher order grammatical types (such as those of adjectives or verbs) constructed from $n$, $s$ and their adjoints by tensor product. For example, an adjective has type $n\cdot n^l$, an intransitive verb has type $n^r\cdot  s$ and a transitive verb $n^r\cdot  s\cdot  n^l$. This grammar can be used to verify if a sentence is well-formed: by means of type reductions, we are able to infer the type of the overall sentence, and check if it corresponds to something meaningful in our interpretation. For example, if a string reduces to the type $s$ then the sentence is judged to be grammatical. The sentence \emph{Clowns tell jokes}, depicted in Figure~\ref{fig:clowns} below, is typed $n\cdot (n^r\cdot s\cdot n^l)\cdot n$, and can be reduced to $s$ as~$n\cdot (n^r\cdot s\cdot n^l)\cdot n \leq 1\cdot s\cdot n^l\cdot n \leq 1 \cdot s \cdot 1 \leq s$. Ultimately, the existence of cups/caps in the pregroup is what makes grammatical reductions possible in this framework.
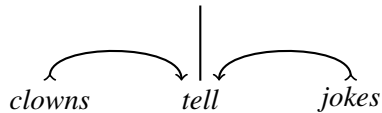
\begin{figure}[h]
\centering
\begin{tikzpicture}[scale=2]
	\begin{pgfonlayer}{nodelayer}
		\node [style=none] (0) at (-2, -0.25) {\emph{clowns}};
		\node [style=none] (1) at (0, -0.25) {\emph{tell}};
		\node [style=none] (2) at (2, -0.25) {\emph{jokes}};
		\node [style=none] (3) at (-2, 0) {};
		\node [style=none] (4) at (-0.25, 0) {};
		\node [style=none] (5) at (0, 0) {};
		\node [style=none] (6) at (0.25, 0) {};
		\node [style=none] (7) at (2, 0) {};
		\node [style=none] (8) at (0, 1) {};
	\end{pgfonlayer}
	\begin{pgfonlayer}{edgelayer}
		\draw [>->, thick, bend left=90, looseness=0.75] (3.center) to (4.center);
		\draw [<-<, thick, bend left=90, looseness=0.75] (6.center) to (7.center);
		\draw [thick] (5.center) to (8.center);
	\end{pgfonlayer}
\end{tikzpicture}
\caption{Graphical reduction of the sentence \emph{Clowns tell jokes}.}\label{fig:clowns}
\end{figure}

Lambek's pregroups provide a categorical model of the \textit{grammar} of natural language, but additional data needs to be specified in order to appropriately capture its \textit{semantics}. Within the DisCoCat framework, this is done by specifying a compact closed symmetric monoidal category $\mathcal{C}$ (the \textit{semantic category}), together with a strong monoidal functor $F$ from the pregroup grammar $P$ into $\mathcal{C}$. Grammatical types $p \in P$ are mapped to objects $F(p)$ of the category (the corresponding \textit{semantic spaces}), and the states of $F(p)$ are taken to provide possible semantics for the fragments of grammatical type $p$ (e.g. words of the different possible types, but also phrases and sentences). Because $F$ is strong monoidal, it sends cups and caps of the pregroup grammar into cups and caps of the semantic category: as a consequence, it transfers compositionality from grammar to semantics, providing a concrete algorithm to compute the \inlineQuote{meaning} of phrases and sentences---understood as a state in the object corresponding to the appropriate grammatical type---from the \inlineQuote{meaning} of their individual constituent words. Borrowing words from General von Clausewitz~\cite{Clausewitz2009}, compact closed symmetric monoidal structure constitutes the real \inlineQuote{\textit{centre of gravity}} of the operational philosophy of categorical compositional distributional semantics.

\paragraph{Significance of our result in linguistics.} Many popular choices of semantic categories take the form $\RMatCategory{S}$ for some commutative semiring $S$, with the involution almost always taken to be trivial: we refer to these as \textit{linear} models of meaning, and we list below the most common examples appearing in the literature. Their presence is so pervasive as to beg the question whether there is some structural reason for their almost complete domination of the DisCoCat scene.

\begin{itemize}
  \item Finite-dimensional real vector spaces are the linear model of meaning par excellence used in natural language processing applications, amounting to the consideration of $\RMatCategory{\reals}$ as semantic category. More generally, one can consider finite-dimensional vector spaces over other fields, such as $\RMatCategory{\complexs}$ (with complex conjugation as involution) or $\RMatCategory{\rationals}$ (with the identity as involution).

  \item The category $\fRelCategory$ of finite sets and relations provides another archetypal example of semantic category for natural language, close in spirit to the formal semantics in the style of Montague. In this context, a semantic space corresponds to some finite set $X$, and the possible semantics (i.e. the states in the space) are given by subsets $U \subseteq X$. The category $\fRelCategory$ is isomorphic to the category $\RMatCategory{\mathbb{B}}$, where $\mathbb{B} := (\{\bot,\top\},\vee,\bot,\wedge,\top)$ is the boolean semiring.

  \item One can also consider more general categories of finite sets and relations, where truth values are taken in a quantale (see~\cite{Marsden2017} for a general treatment of these models of semantics). Every quantale $Q$ is a commutative semiring $(Q,\vee,\bot,\times,1)$, and using finite sets and relations over generalised truth values amounts to the consideration of $\RMatCategory{Q}$ as semantic category.

  \item The category of sets and histograms, corresponding to $\RMatCategory{\naturals}$, has been recently proposed as a counting-based model of semantics~\cite{Gogioso2016}. In this model, a semantic space is given by a finite set $X$, and meaning of words/phrases is given by histograms over $X$. 

\end{itemize}
Keeping these examples in mind, the uniqueness result of Theorem \ref{thm_freeFdSemimodulesRlinearUniqueness} equates to a pretty strong statement about the relationship between grammar and semantics in distributional models of natural language: {\it given a linear model of meaning, the only real freedom in specifying a grammatical structure is the choice of semantic spaces for the generators of the grammar}. This confirms what has been folklore in the DisCoCat community for some time now, and an unstated assumption of many works in its ecosystem: categories in the form $\RMatCategory{S}$ are so common because there is essentially no other way of categorically combining compositional grammar with the underlying linear structure of semantics.

\vspace{-2mm}

\section{Conclusions and Future Work}
\label{section_conclusions}

As part of this work, we have proven that compact closed symmetric monoidal structure is a property---rather than merely a structure---when it comes to the process theories commonly used in the categorical study of quantum theory and natural language processing. From the point of view of quantum theory, this means that there is an essentially unique notion of composite quantum systems and maximal entanglement. From the point of view of natural language processing, this means that linear models of semantics come equipped with a natural choice of compositional grammatical structure on them. In the future, we will be interested in obtaining similar essential uniqueness results for larger classes of process theories, such as those involved in the treatment of infinite-dimensional quantum systems, and some infinite-dimensional models of semantics for natural language processing. 

\vspace{-3.5mm}

\subparagraph*{Acknowledgements.}
SG gratefully acknowledges funding from EPSRC and Trinity College. DM is funded by the AFSOR grant ``Algorithmic and Logical Aspects when Composing Meanings'' and FQXi grant ``Categorical Compositional Physics''. FG and RP are funded by the AFSOR grant ``Algorithmic and Logical Aspects when Composing Meanings''. This project/publication was made possible through the support of a grant from the John Templeton Foundation. The opinions expressed in this publication are those of the authors and do not necessarily reflect the views of the John Templeton Foundation.

\vspace{-3.5mm}


\begin{thebibliography}{10}
\providecommand{\bibitemdeclare}[2]{}
\providecommand{\surnamestart}{}
\providecommand{\surnameend}{}
\providecommand{\urlprefix}{Available at }
\providecommand{\url}[1]{\texttt{#1}}
\providecommand{\href}[2]{\texttt{#2}}
\providecommand{\urlalt}[2]{\href{#1}{#2}}
\providecommand{\doi}[1]{doi:\urlalt{http://dx.doi.org/#1}{#1}}
\providecommand{\bibinfo}[2]{#2}

\bibitemdeclare{inproceedings}{AC1}
\bibitem{AC1}
\bibinfo{author}{S.~\surnamestart Abramsky\surnameend} \&
  \bibinfo{author}{B.~\surnamestart Coecke\surnameend} (\bibinfo{year}{2004}):
  \emph{\bibinfo{title}{A categorical semantics of quantum protocols}}.
\newblock In: {\sl \bibinfo{booktitle}{Proceedings of the 19th Annual IEEE
  Symposium on Logic in Computer Science (LICS)}}, pp.
  \bibinfo{pages}{415--425},
  \doi{10.1109/LICS.2004.1319636}

\bibitemdeclare{article}{Clausewitz2009}
\bibitem{Clausewitz2009}
\bibinfo{author}{Carl~von \surnamestart Clausewitz\surnameend}
  (\bibinfo{year}{2009}): \emph{\bibinfo{title}{On War: the complete edition}}.
\newblock {\sl \bibinfo{journal}{Wildside Press}}.

\bibitemdeclare{book}{CKBook}
\bibitem{CKBook}
\bibinfo{author}{B.~\surnamestart Coecke\surnameend} \&
  \bibinfo{author}{A.~\surnamestart Kissinger\surnameend}
  (\bibinfo{year}{2017}): \emph{\bibinfo{title}{Picturing Quantum Processes. A
  First Course in Quantum Theory and Diagrammatic Reasoning}}.
\newblock \bibinfo{publisher}{Cambridge University Press},
\doi{10.1017/9781316219317}.

\bibitemdeclare{article}{Coecke2010}
\bibitem{Coecke2010}
\bibinfo{author}{B.~\surnamestart Coecke\surnameend},
  \bibinfo{author}{M.~\surnamestart Sadrzadeh\surnameend} \&
  \bibinfo{author}{S.~\surnamestart Clark\surnameend} (\bibinfo{year}{2010}):
  \emph{\bibinfo{title}{Mathematical Foundations for Distributed Compositional
  Model of Meaning. {L}ambek Festschrift}}.
\newblock {\sl \bibinfo{journal}{Linguistic Analysis}} \bibinfo{volume}{36},
  pp. \bibinfo{pages}{345--384}.

\bibitemdeclare{article}{Coecke2013}
\bibitem{Coecke2013}
\bibinfo{author}{Bob \surnamestart Coecke\surnameend}, \bibinfo{author}{Edward
  \surnamestart Grefenstette\surnameend} \& \bibinfo{author}{Mehrnoosh
  \surnamestart Sadrzadeh\surnameend} (\bibinfo{year}{2013}):
  \emph{\bibinfo{title}{Lambek vs. Lambek: Functorial vector space semantics
  and string diagrams for Lambek calculus}}.
\newblock {\sl \bibinfo{journal}{Annals of pure and applied logic}}
  \bibinfo{volume}{164}(\bibinfo{number}{11}), pp. \bibinfo{pages}{1079--1100},
  \doi{10.1016/j.apal.2013.05.009}.

\bibitemdeclare{inproceedings}{Gogioso2016}
\bibitem{Gogioso2016}
\bibinfo{author}{Stefano \surnamestart Gogioso\surnameend}
  (\bibinfo{year}{2016}): \emph{\bibinfo{title}{A Corpus-based Toy Model for
  DisCoCat}}.
\newblock In \bibinfo{editor}{Dimitrios \surnamestart Kartsaklis\surnameend},
  \bibinfo{editor}{Martha \surnamestart Lewis\surnameend} \&
  \bibinfo{editor}{Laura \surnamestart Rimell\surnameend}, editors: {\sl
  \bibinfo{booktitle}{{\rm Proceedings of the 2016 Workshop on} Semantic Spaces
  at the Intersection of NLP, Physics and Cognitive Science, {\rm Glasgow,
  Scotland, 11th June 2016}}}, {\sl \bibinfo{series}{Electronic Proceedings in
  Theoretical Computer Science}} \bibinfo{volume}{221},
  \bibinfo{publisher}{Open Publishing Association}, pp.
  \bibinfo{pages}{20--28}, \doi{10.4204/EPTCS.221.3}.

\bibitemdeclare{article}{Gogioso2017}
\bibitem{Gogioso2017}
\bibinfo{author}{Stefano \surnamestart Gogioso\surnameend}
  (\bibinfo{year}{2017}): \emph{\bibinfo{title}{Fantastic Quantum Theories and
  Where to Find Them}}.
\newblock {\sl \bibinfo{journal}{\href{https://arxiv.org/abs/1703.10576}{arXiv:1703.10576}}}.

\bibitemdeclare{inproceedings}{Grefenstette2011}
\bibitem{Grefenstette2011}
\bibinfo{author}{E.~\surnamestart Grefenstette\surnameend} \&
  \bibinfo{author}{M.~\surnamestart Sadrzadeh\surnameend}
  (\bibinfo{year}{2011}): \emph{\bibinfo{title}{Experimental Support for a
  Categorical Compositional Distributional Model of Meaning}}.
\newblock In: {\sl \bibinfo{booktitle}{The 2014 Conference on Empirical Methods
  on Natural Language Processing.}}, pp. \bibinfo{pages}{1394--1404}.

\bibitemdeclare{inproceedings}{Kartsaklis2014}
\bibitem{Kartsaklis2014}
\bibinfo{author}{Dimitri \surnamestart Kartsaklis\surnameend} \&
  \bibinfo{author}{Mehrnoosh \surnamestart Sadrzadeh\surnameend}
  (\bibinfo{year}{2014}): \emph{\bibinfo{title}{A Study of Entanglement in a
  Categorical Framework of Natural Language}}.
\newblock In \bibinfo{editor}{Bob \surnamestart Coecke\surnameend},
  \bibinfo{editor}{Ichiro \surnamestart Hasuo\surnameend} \&
  \bibinfo{editor}{Prakash \surnamestart Panangaden\surnameend}, editors: {\sl
  \bibinfo{booktitle}{{\rm Proceedings of the 11th workshop on} Quantum Physics
  and Logic, {\rm Kyoto, Japan, 4-6th June 2014}}}, {\sl
  \bibinfo{series}{Electronic Proceedings in Theoretical Computer Science}}
  \bibinfo{volume}{172}, \bibinfo{publisher}{Open Publishing Association}, pp.
  \bibinfo{pages}{249--261}, \doi{10.4204/EPTCS.172.17}.

\bibitemdeclare{incollection}{Kelly}
\bibitem{Kelly}
\bibinfo{author}{G.~M. \surnamestart Kelly\surnameend} (\bibinfo{year}{1972}):
  \emph{\bibinfo{title}{Many-variable functorial calculus {I}}}.
\newblock In \bibinfo{editor}{G.~M. \surnamestart Kelly\surnameend},
  \bibinfo{editor}{M.~\surnamestart Laplaza\surnameend},
  \bibinfo{editor}{G.~\surnamestart Lewis\surnameend} \&
  \bibinfo{editor}{S.~Mac \surnamestart Lane\surnameend}, editors: {\sl
  \bibinfo{booktitle}{Coherence in Categories}}, {\sl \bibinfo{series}{Lecture
  Notes in Mathematics}} \bibinfo{volume}{281},
  \bibinfo{publisher}{Springer-Verlag}, pp. \bibinfo{pages}{66--105},
  \doi{10.1007/BF01898420}.

\bibitemdeclare{article}{KellyLaplaza}
\bibitem{KellyLaplaza}
\bibinfo{author}{G.~M. \surnamestart Kelly\surnameend} \&
  \bibinfo{author}{M.~L. \surnamestart Laplaza\surnameend}
  (\bibinfo{year}{1980}): \emph{\bibinfo{title}{Coherence for compact closed
  categories}}.
\newblock {\sl \bibinfo{journal}{Journal of Pure and Applied Algebra}}
  \bibinfo{volume}{19}, pp. \bibinfo{pages}{193--213},
  \doi{10.1016/0022-4049(80)90101-2}.

\bibitemdeclare{inproceedings}{Lambek1997}
\bibitem{Lambek1997}
\bibinfo{author}{Joachim \surnamestart Lambek\surnameend}
  (\bibinfo{year}{1997}): \emph{\bibinfo{title}{Type grammar revisited}}.
\newblock In: {\sl \bibinfo{booktitle}{International Conference on Logical
  Aspects of Computational Linguistics}}, \bibinfo{organization}{Springer}, pp.
  \bibinfo{pages}{1--27},
  \doi{10.1007/3-540-48975-4\_1}

\bibitemdeclare{misc}{Marsden2017}
\bibitem{Marsden2017}
\bibinfo{author}{Dan \surnamestart Marsden\surnameend} \&
  \bibinfo{author}{Fabrizio \surnamestart Genovese\surnameend}
  (\bibinfo{year}{2017}): \emph{\bibinfo{title}{Custom Hypergraph Categories
  via Generalized Relations}}.
\newblock {\sl \bibinfo{journal}{\href{https://arxiv.org/abs/1703.01204}{arXiv:1703.01204}}}.

\bibitemdeclare{misc}{Selinger2009}
\bibitem{Selinger2009}
\bibinfo{author}{Peter \surnamestart Selinger\surnameend}
  (\bibinfo{year}{2009}): \emph{\bibinfo{title}{A survey of graphical languages for monoidal categories}}.
\newblock {\sl \bibinfo{journal}{Lecture Notes in Physics}}
  \bibinfo{volume}{813}, pp. \bibinfo{pages}{289--355}.

\end{thebibliography}

\newpage
\appendix

\end{document}